\documentclass[twoside]{article}

\usepackage[accepted]{aistats2021}
\usepackage[round,authoryear]{natbib}

\usepackage[usenames]{xcolor}
\definecolor{webgreen}{rgb}{0,.5,0}
\definecolor{webbrown}{rgb}{.6,0,0}
\usepackage[colorlinks=true,
linkcolor=webgreen,
urlcolor=blue,
filecolor=webbrown,
citecolor=webgreen]{hyperref}

\usepackage{amsfonts}
\usepackage{amsthm}
\usepackage{tikz}
\usepackage{url}
\usepackage{relsize}
\usepackage{enumitem}
\usepackage{algorithm}
\usepackage[noend]{algpseudocode}
\usepackage{tcolorbox}

\theoremstyle{plain}
\newtheorem{theorem}{Theorem}
\newtheorem*{theorem*}{Theorem}

\newtheorem*{corollary*}{Corollary}
\newtheorem{lemma}[theorem]{Lemma}
\newtheorem*{lemma*}{Lemma}
\newtheorem{proposition}[theorem]{Proposition}
\newtheorem*{proposition*}{Proposition}

\theoremstyle{definition}

\newtheorem*{definition*}{Definition}

\newtheorem*{example*}{Example}

\newtheorem*{conjecture*}{Conjecture}

\newtheorem*{problem*}{Problem}

\newtheorem*{todo*}{TO DO}

\newtheorem*{question*}{Question}

\theoremstyle{remark}

\newtheorem*{remark*}{Remark}

\newtheorem*{notation*}{Notation}

\numberwithin{theorem}{section}

\begin{document}

%
\runningtitle{\MakeLowercase{{independence:}} Fast Rank Tests}

%

\twocolumn[

\aistatstitle{\MakeLowercase{\LARGE\texttt{independence:}} Fast Rank Tests}

\aistatsauthor{ Chaim Even-Zohar }

\aistatsaddress{ The Alan Turing Institute } ]

\begin{abstract}
In 1948 Hoeffding devised a nonparametric test that detects dependence between two continuous random variables $X$ and $Y$, based on the ranking of $n$ paired samples $(X_i,Y_i)$.  The computation of this commonly-used test statistic takes $O(n \log n)$ time. Hoeffding's test is consistent against any dependent probability density $f(x,y)$, but can be fooled by other bivariate distributions with continuous margins.  Variants of this test with full consistency have been considered by Blum, Kiefer, and Rosenblatt (1961), Yanagimoto (1970), Bergsma and Dassios (2010). The so far best known algorithms to compute these stronger independence tests have required quadratic time. Here we improve their run time to $O(n \log n)$, by elaborating on new methods for counting ranking patterns, from a recent paper by the author and Leng (SODA'21). Therefore, in all circumstances under which the classical Hoeffding independence test is applicable, we provide novel competitive algorithms for consistent testing against all alternatives. Our \texttt{R}~package, \texttt{independence}, offers a highly optimized implementation of these rank-based tests. We demonstrate its capabilities on large-scale datasets.
\end{abstract}

\section{OVERVIEW}

Part~\ref{setting} briefly outlines the current need for consistent, fast, and distribution-free independence testing and measuring.

Part~\ref{hoef} surveys the original Hoeffding's $D$ measure of dependence \citep{hoeffding1948non}. Since it is solely based on ranks, the derived test statistic is indeed distribution-free. 

Hoeffding's test is only consistent against absolutely continuous bivariate distributions. Various examples are given in Part~\ref{fooling} of simple dependent alternatives that Hoeffding's $D$ fails to detect.

Part~\ref{variations} presents the two main refinements of Hoeffding's dependence measure, that yield consistent distribution-free tests against all alternatives. They go back to works of \cite{blum1961distribution}, \cite{yanagimoto1970measures}, and others. Recent works by \cite{bergsma2014consistent}, and others have popularized and extended these variants of Hoeffding's test. 

Part~\ref{algo} gives $O(n \log n)$ time algorithms of these consistent variants of Hoeffding's independence test, based on techniques of \cite{even2019counting}. Such an improvement from quadratic to near linear time has been regarded as the most important factor in selecting an independence test \citep{chatterjee2020new}, and may be considered as overcoming a natural barrier \citep[cf.][]{deheuvels2006quadratic}. By designing a dedicated algorithm for these statistics, we achieve an additional \mbox{5x-10x} improvement over direct application of our general method.

Part~\ref{rpackage} briefly describes the functionality of our newly developed \texttt{R} package \texttt{independence}
\citep[available on \texttt{CRAN},][]{independenceR},
and demonstrates its unique capabilities.

Finally, in Part~\ref{related} we briefly survey other ranking-based independence tests, that have been proposed within the same setting of Hoeffding's test. 

\section{SETTING}\label{setting}

Analyzing the empirical relationship between two variables is a classical subject in statistics. Perhaps the most basic problem given bivariate data is to decide whether there is any relation at all between the two variables. It is therefore natural that \emph{testing for independence} has long been a primary interest of statistical researchers and practitioners.

One crucial feature is the test's \emph{consistency}. Traditional tests of linear correlation, such as Pearson's $r$, Spearman's $\rho$, or Kendall's $\tau$, may rule out the independence hypothesis in some cases. However, they are not consistent, since many forms of dependency are uncorrelated, and hence undetectable by them. In general, one seeks consistent tests, that assume as little as possible on the underlying distribution, and hence eliminate many alternatives given enough data.

Another key feature is the \emph{computational complexity} of the proposed test. In the times of ``big data'', we are able to discover latent connections and subtle patterns, that are often detectable only with a great number of samples. In a typical scenario of, say, several million data records, and testing potential relations between tens of attributes, even the gap between linear and quadratic runtime is a matter of vital importance.

This work provides \emph{consistent and efficient} independence tests for $n$ paired samples of two real-valued variables $X$ and~$Y$. The variables are only assumed to be non-atomic, and the tests are \emph{distribution-free}. This is the classical setting of Hoeffding's $D$ and its variants, where one has so far been compelled to choose between consistency and efficiency. 

For other, simpler data types, consistent linear-time methods are already widely available, such as chi-square tests for categorical variables. On the contrary, testing for independence in more involved settings, such as real vectors and processes, is an active area of research. The rank-based tests that we implement here extend to more general settings in various ways, and it is plausible that our approach can help accelerate and optimize some of these generalized methods.

The implemented tests are \emph{nonparametric}, as no assumptions are made on the distribution of $X$ or $Y$, other than continuity which is normally evident from the nature of the data. Given $n$ independent observations, $\{(X_1,Y_1),\dots,(X_n,Y_n)\}$, we only make use of the relative ranking of the two coordinates, which is invariant to any monotone reparametrization of the variables. This information may be encoded by the rank-matching permutation: $\pi : \{1,\dots,n\} \to \{1,\dots,n\}$ where $\pi(\textrm{rank\,}X_i) = \textrm{rank\,}Y_i$. Since $\pi$ is uniformly random under independence, the null distribution only depends on~$n$, and conveniently one can use distribution-free $p$-values.

\section{HOEFFDING'S TEST}\label{hoef}

We first review Hoeffding's independence test, its usefulness, and limitations. Our presentation mostly follows the seminal paper by \cite{hoeffding1948non}. 

\subsection{Hoeffding's Independence Coefficient}

Let $X$ and $Y$ be two real-valued random variables on the same probability space, with a joint cumulative distribution function $F(x,y) = P(X \leq x, Y \leq y)$. Their marginal distribution functions are assumed to be continuous, and are denoted respectively by 
\begin{eqnarray*}
&F_X(x) = F(x, \infty) = P(X \leq x) \\
&F_Y(y) = F(\infty, y) = P(Y \leq y)
\end{eqnarray*}
The null hypothesis, that the two variables are independent, is expressed as: $F \equiv F_XF_Y$. 

In order to measure how much $F$ departs from independence, Hoeffding proposed the following Cram\'er--von Mises type statistic.
$$ D \;=\; \int\left(F(x,y) - F_X(x)F_Y(y)\right)^2dF(x,y) $$ 
The integration $dF$ over $\mathbb{R}^2$ can be interpreted in the sense of Riemann--Stieltjes or Lebesgue--Stieltjes. This notation emphasizes the repeated use of the same $F$ both for evaluating the deviation and for averaging it.

It is clear that $D=0$ if $X$ and $Y$ are independent. Conversely, \citet[Theorem~3.1]{hoeffding1948non} showed that $D>0$ whenever the joint distribution is dependent and \emph{absolutely} continuous, that is, has a joint density function $f(x,y)$ integrating to~$F(x,y)$. 

Next, we briefly sketch Hoeffding's derivation of an estimator for~$D$, since it will be useful later on. Expanding the square, we write
\begin{equation}
\label{dfffff}
D \,=\, \int\left(F\,F - 2\,F\,F_XF_Y + F_XF_YF_XF_Y\right)dF
\tag{$\star$}
\end{equation}
The cumulative distribution functions are then replaced by integrals,
\begin{align*}
\textstyle
F(x,y) \;&=\; \int_{\mathbb{R}^2} I(x'<x,y'<y)\,dF(x',y') \\
F_X(x) \;&=\; \int_{\mathbb{R}^2} I(x'<x)\,dF(x',y') \\
1 \;&=\; \int_{\mathbb{R}^2} dF(x',y')
\end{align*}
where $I(\cdots)=1$ if the given condition holds and~$0$ otherwise. We end up with an expression of the following form.
$$ D \;=\; \int\!\!\int\!\!\int\!\!\int\!\!\int \phi\left(x_1,y_1,\dots,x_5,y_5\right) dF\,dF\,dF\,dF\,dF$$
The function $\phi:(\mathbb{R}^2)^{5} \to \mathbb{R}$ can be written explicitly, and only depends on the order relations between its inputs, and not on the distribution~$F$. 

\subsection{Hoeffding's Test Statistic}
\label{hoefdn}

Recall that we are given $n$ independent observations, $(X_i,Y_i)$. Hoeffding's test statistic is essentially obtained by using the empirical distribution instead of the unknown~$F$. Specifically, the integrals turn into a finite average:
\begin{equation}\label{dn}\tag{$\diamond$} 
D_n \;=\;\; \frac{\displaystyle{\sum}_{i,j,k,l,m} \phi\left(X_i,Y_i,\dots,X_m,Y_m\right)}{n(n-1)(n-2)(n-3)(n-4)} 
\end{equation}
The sum is over all combinations of five distinct indices $i,j,k,l,m \in \{1,\dots,n\}$. Estimators obtained this way are called U-statistics, and were developed and studied by \cite{hoeffding1948class}. Such an estimator is unbiased, $E(D_n)=D$, and also consistent, $D_n \to D$ in probability, where $D$ is the population parameter as above.

This expression for $D_n$ can be interpreted in terms of rank patterns, in the sense of the following definitions. Recall that any set of points $\left\{(x_1,y_1),\dots,(x_n,y_n)\right\}$ induces a permutation $\pi \in S_n$ such that $\pi(\textrm{rank\,}x_i) = \textrm{rank\,}y_i$. We use the one-line notation $\pi(1)\pi(2)\dots\pi(n)$ for permutations, for example $\pi = 1423$. The number of $k$-point subsets inducing a given \emph{pattern} $\sigma \in S_k$ is denoted by $\#\sigma(\pi)$, or $\#\sigma$ for short. For example $\#21(1423)=2$, since the pattern $21$ occurs twice, at the marked entries $1\underline{4}\underline{2}3$ and~$1\underline{4}2\underline{3}$. 

The statistic $D_n$ can now be represented by counting patterns of order 5 in the permutation induced by $(X_1,Y_1),\dots,(X_n,Y_n)$. Indeed, by a straightforward case analysis of the kernel $\phi$, the sum in Equation~(\ref{dn}) equals
\begin{multline*}
4\left(\#12345 + \#12354 + \dots + \#54321\right) \\ \;-\; 2\left(\#14325 + \#14352 + \dots + \#52341\right)    
\end{multline*}
This combination involves all occurrences of 5-patterns whose middle entry is~3. The eight patterns where it separates 1 and 2 from 4 and 5 are being added, and the other sixteen patterns are subtracted.

For every point $(X_i,Y_i)$, let $a_i$, $b_i$, $c_i$ and $d_i$ be the number of other points in the four quadrants around it, as follows.
\begin{align*}
a_i = \left|\left\{j\;:\;\substack{X_j<X_i\\Y_j>Y_i}\right\}\right| && 
b_i = \left|\left\{j\;:\;\substack{X_j>X_i\\Y_j>Y_i}\right\}\right| \\
c_i = \left|\left\{j\;:\;\substack{X_j<X_i\\Y_j<Y_i}\right\}\right| &&
d_i = \left|\left\{j\;:\;\substack{X_j>X_i\\Y_j<Y_i}\right\}\right|
\end{align*}
\cite{hoeffding1948non} gives a formula for computing~$D_n$, which can be rephrased by means of these variables as follows:
\begin{multline*}
D_n \;=\; \tfrac{1}{n(n-1)\cdots(n-4)} \sum\limits_{i=1}^{n} \scalebox{1.2}{[} a_i(a_i-1)d_i(d_i-1) \\ \;+\; b_i(b_i-1)c_i(c_i-1) \;-\; 2\,a_i b_i c_i d_i \scalebox{1.2}{]}
\end{multline*}
\cite{blum1961distribution} studied Hoeffding's statistic, and proposed a slightly more convenient estimator:
$$ B_n \;=\; \frac{1}{n^5}\sum\limits_{i=1}^n\,\left[a_id_i-b_ic_i\right]^2 $$ 
Asymptotically in~$n$, these two statistics differ by a constant offset to the normalized null distribution, and a negligible term under continuous alternatives.

The sequence $\{a_i\}$ is essentially the \emph{inversion code} of the induced permutation~$\pi$ \cite[p.~36]{stanley2011enumerative}. The computation of this sequence, and thereby the other three, in $O(n \log n)$ time is a textbook exercise, see~\cite[14.1-7, for example]{cormen2009introduction} or the tools we use in~\S\ref{algo}. Hence the computation of Hoeffding's~$D_n$ costs $O(n \log n)$ time.

\subsection{Hoeffding's Independence Test}
\label{htest}

As usual, Hoeffding's test is performed by computing $D_n$ and comparing it with its distribution under the independence hypothesis. If $D_n$ is significantly large then independence is reject. This happens with probability tending to one under any alternative with $D>0$.

The null distribution is uniquely determined regardless of $F_X$ and $F_Y$, since any ordering for the $X_i$'s is equally likely, and same for the~$Y_i$'s. For large $n$, the distribution is approximated using the following limit law~\cite[Theorem 8.1]{hoeffding1948non}.
$$ n \, D_n \;\;\; \xrightarrow[\;\; n \to \infty \;\;]{} \;\;\;  \sum_{j=1}^\infty\sum_{k=1}^\infty \frac{1}{\pi^4j^2k^2}\left(Z_{jk}^2 - 1\right) $$
where $Z_{jk}$ are independent gaussians. See also~\cite{blum1961distribution}. For further information on the distribution of general linear combinations of pattern counts in a random permutation see~\cite{janson2015asymptotic} and \cite{even2018patterns}. 

Efficient implementations of Hoeffding's~$D$ test are available in popular software, e.g.~\citet[\href{https://cran.r-project.org/package=wdm}{\texttt{wdm} package}]{R}, \citet[\href{https://reference.wolfram.com/language/ref/HoeffdingD.html}{\texttt{HoeffdingD}}]{Mathematica}, \citet[\href{https://v8doc.sas.com/sashtml/proc/zconcept.htm}{\texttt{PROC CORR}}]{sas2015base}.

\begin{remark*}
Although our discussion mainly adopts the standpoint of hypothesis testing, Hoeffding's $D_n$ is also much relevant from a descriptive point of view. It is meaningful as a distribution-free measure of dependence, consistently estimating the population's $D \in [0,\tfrac{1}{30}]$. Such a quantity may serve as a distance function in the context of \emph{variable clustering}, for example, and Hoeffding's~$D_n$ in particular has been recommended in these pursuits, see \cite{harrell2015regression}.
\end{remark*}

\section{FOOLING HOEFFDING}
\label{fooling}

As mentioned above, the consistency of Hoeffding's test was established under the condition that the joint distribution of $(X,Y)$ is absolutely continuous. It is not guaranteed that $D>0$ for other dependent alternatives. Some distributions with no density still happen to be detectable. For example, $D=\tfrac{1}{30}$ in the important case of a monotone dependency, $Y=\varphi(X)$ almost surely. However, $D$~might vanish under other dependent distribution. We discuss various examples.

\paragraph{(I)}

The first example was given by \cite{yanagimoto1970measures}. It can be described as a perturbation of an independent distribution. Without loss of generality we start from the uniform measure, $U([0,1]\times[0,1])$. Inside this square we choose a parallelogram made of two axes-parallel right triangles: \raisebox{-0.1em}{\tikz[scale=0.125]{\draw[black,line width=1] (0,0)--(3,2)--(0,2)--(-3,0)--(0,0)--(0,2);}}. The probability mass in the interior of the triangles is then transferred to the diagonal edges, upward in one triangle and downward in the other. Clearly, this makes the coordinates $X$ and $Y$ dependent, while still non-atomic. See Figure~\ref{fool}, and the paper of \cite{yanagimoto1970measures}.

\begin{figure}[t]
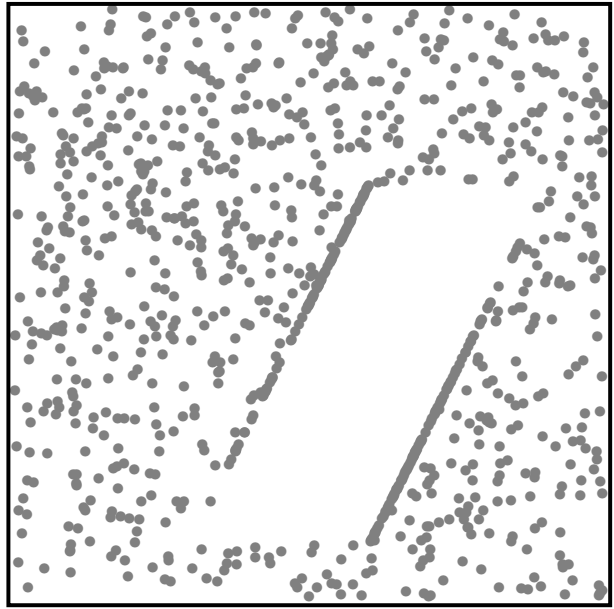

\centering
\tikz[scale=8]{
\foreach[evaluate={
\x = 0.49*rand+0.5;
\y = 0.49*rand+0.5;
\y = (\x>0.3 && \x<0.6 && \y>0.1 && \y<2*\x-0.5?2*\x-0.5:\y;
\y = (\x>=0.6 && \x<0.9 && \y<0.7 && \y>2*\x-1.1?2*\x-1.1:\y;
}] \i in {1,...,1000}{ 
\node[gray,mark size=1pt] at (\x,\y) {$\bullet$};}
\draw[black, line width=1.5] (0,0) rectangle (1,1);
}
\vspace{.3in}
\caption{A Random Sample of 1000 Points from Yanagimoto's Counterexample.}
\vspace{.1in}
\label{fool}
\end{figure}

This distribution satisfies $D=0$. Indeed, $F_X(x)=x$ for $x \in [0,1]$ because the probability mass is only ``swept'' vertically. Also $F_Y(y)=y$ for $y \in [0,1]$ as the two diagonals complement each other. Finally, $F=F_XF_Y$ almost everywhere with respect to~$F$. Indeed, let $(x,y)$ be a point in the support of this distribution. At least one of the four quadrants around it is unchanged, so  $F(x,y)=xy$ follows.
Of course this is not the case in the interior of the parallelogram, but that area is invisible to the integral~$dF$.

\paragraph{(II)}

One can observe that Yanagimoto's counterexample is actually quite robust. The position, size and ratio of the parallelogram do not matter. Several disjoint parallelograms would work as well. 

Moreover, the two diagonal lines may be replaced by other monotone curves that preserve the marginals. For example, one may sweep all the in-between probability mass to the following two hyperbola branches:
$$ \begin{array}{c} 2y(1-x)=1 \\[0.25em] 2x(1-y)=1 \end{array} \;\;\;\;\;\;\;\; \raisebox{-1.125em}{\tikz[scale=1]{ \fill[draw=black,fill=lightgray,line width=0.7] (0,1)--(0,0.5)--(0.050,0.526)--(0.100,0.556)--(0.150,0.588)--(0.200,0.625)--(0.250,0.667)--(0.300,0.714)--(0.350,0.769)--(0.400,0.833)--(0.450,0.909)--(0.5,1)--(0,1) (1,0)--(0.5,0)--(0.526,0.050)--(0.556,0.100)--(0.588,0.150)--(0.625,0.200)--(0.667,0.250)--(0.714,0.300)--(0.769,0.350)--(0.833,0.400)--(0.909,0.450)--(1,0.5)--(1,0); \draw[lightgray,line width=0.8] (0,0)--(1,0)--(1,1)--(0,1)--(0,0);}} $$
This clears an area of $\log 2 \approx 0.69$ in the unit square while $F_X(x)=x$ and~$F_Y(y)=y$ for all $x,y \in [0,1]$. Moreover, splitting this distribution at the median of~$X$ completely coincides with the splitting by~$Y$. Still, Hoeffding's $D=0$ for the same reasons as before.

\paragraph{(III)}

We give another new counterexample, of a different nature. It is based on the binary expansion $X = 0.X_1X_2X_3\dots$ where $X_i$ are independent Bernoulli$(\tfrac12)$ random variables. Denote by $\ell = \ell(X) \in \mathbb{N}$ the length of the first run in this binary representation, namely $X_1 = X_2 = \dots = X_{\ell} \neq X_{\ell+1}$. Let $Y = 0.Y_1Y_2Y_3\dots$ be obtained by resampling the first $\ell$ digits of~$X$, so that $Y_1,\dots,Y_{\ell}$ and $X$ are independent given~$\ell$, while $Y_i = X_i$ for $i > \ell$. Here is an example, with the first run marked.
\begin{align*}
X \;&=\; \texttt{ 0.\underline{00000}1101000101...} \\
Y \;&=\; \texttt{ 0.\underline{01101}1101000101...}
\end{align*}

Each of the two variables $X$ and $Y$ follows a continuous uniform distribution on the interval $[0,1]$, and their joint distribution is clearly dependent. However, we show that these variables fool Hoeffding's test. 

\begin{proposition}\label{fool2}
For $X$ and $Y$ distributed as above, Hoeffding's $D = 0$.
\end{proposition}

The proof is given in the supplementary material, \S\ref{foolproof}.

The dependency in this counterexample is stronger than the previous one in the sense that for any given $X$ there are only finitely many options for~$Y$, and often as few as two. We note that the proof works the same even if we restrict it to $Y<0.5$. Then the distribution satisfies $Y = \phi(X)$ with probability $2/3$, where $\phi(x) = (x \bmod 0.5)$, yet Hoeffding's $D=0$.

It is an interesting phenomena that the resulting probability measure is \emph{singular} with respect to the independent distribution with the same marginals. Even so, Hoeffding's statistic cannot tell them apart.

We argue that the occurrence of this kind of dependency in a real-world application is not unimaginable. In numerical simulations, for example, it might happen that two outcomes share certain parts of their binary expansion. 

The above examples can be modified to guarantee that the dependency stays undetectable also by standard correlation tests, such as Spearman's~$\rho$ and Kendall's~$\tau$. One way to do that is by adding a ``mirror image'', i.e., taking $X' = \pm X \in [-1,1]$, where the sign is independent of $(X,Y)$ and determined by flipping a fair coin. 

\section{VARIATIONS}\label{variations}

Many variations of 
Hoeffding's test and extensions in variaous directions have been proposed over the years. Here we describe the two main variants that maintain the original spirit of the test while achieving consistency for all non-atomic real-valued $X$ and~$Y$.

\subsection{Refined Hoeffding Test}

The above counterexamples highlight where the blind spots of Hoeffding's test are. The integral~$dF$ that appears in $D$'s definition might miss certain regions in the plane, where $F \neq F_X F_Y$. From this perspective, it is natural to try a double integral $dF_X dF_Y$ instead, which leads to the following modification of~$D$.
$$ R \;=\; \iint\left(F(x,y) - F_X(x)F_Y(y)\right)^2dF_X(x)dF_Y(y) $$
This measure of independence is often attributed to Blum, Kiefer, and Rosenblatt \citep{blum1961distribution}. In their paper, they discuss several possible variants and extensions of Hoeffding's test, and this one is defined as $\gamma_F^2$ on page~490. Curiously, a similar measure had already been defined in an early work of \cite{hoffding1940masstabinvariante}, see \citep[$\Phi^2$~on page~62]{hoeffding1994collected}. It was later introduced elsewhere, such as in the work of \citet[page 58]{yanagimoto1970measures}. 

The parameter $R$ vanishes if and only if $X$ and $Y$ are independent \cite[bottom of page~490]{blum1961distribution}. Therefore, a consistent estimator of~$R$ yields a consistent independence test against all alternatives. We hence derive a test statistic $R_n$ given $n$ independent samples, similar to $D_n$ above. We first expand,
\begin{equation}
\label{bfffff}
R \,=\, \iint\left(F\,F - 2\,F\,F_XF_Y + F_XF_YF_XF_Y\right)\,dF_XdF_Y
\tag{$\star\star$}
\end{equation} 

The integral of the last term separates and simplifies to $\int F_X^2dF_X\int F_Y^2dF_Y = 1/9$. The other terms are replaced with integrals over indicator functions, as for $D$ in (\ref{dfffff}) above. We end up with a fivefold integral $\int\!\!\int\!\!\int\!\!\int\!\!\int \psi \, dFdFdFdFdF$, for some $\psi$ that only depends on the rank patterns induced by its five input points.

Given $(X_1,Y_1), \dots, (X_n,Y_n)$, the estimator $R_n$ may be defined by a sum over sets of five samples, similar to Equation~(\ref{dn}), with $\phi$ replaced by~$\psi$.

The independence test works as described in Section~\ref{htest}, using $R_n$ instead. It turns out that assuming independence $D_n$ and $R_n$ are equivalent up to an order $n^{-3/2}$ term \citep[for example]{even2018patterns}.
In particular, $nD_n$ and $nR_n$ follow the same asymptotic null distribution, described by \cite{hoeffding1948non} and \cite{blum1961distribution}. In Part~\ref{algo} we give a first near linear time algorithm for computing~$R_n$.

\subsection{Bergsma--Dassios--Yanagimoto Test}
\label{bdy} 

The next step in the evolution of Hoeffding's independence test originates in multiple works, including the following question on the distribution of permutation patterns. 

Recall that the permutation induced by $k$ samples $(X_1,Y_1),\dots,(X_k,Y_k)$ is uniformly distributed under the independence hypothesis, with each $\pi \in S_k$ equally likely. Hoeffding's $D$, and its stronger variant $R$, provide a converse statement: if every 5-pattern $\pi \in S_5$ is equally likely to be induced by $(X_1,Y_1),\dots,(X_5,Y_5)$ then $X$ and $Y$ are independent. 

\cite{yanagimoto1970measures} showed that already equiprobable 4-patterns  imply independence, but 3-patterns are not sufficient. This implication is shown by taking a nonnegative combination of the nonnegative dependence measures $D$ and $R$,  
$$ T = D + 2 R $$
Since $D=R=0$ if and only if $X$ and $Y$ are independent, also $T=0$ exactly in that case. Yanagimoto observed that, with this choice of combination, the order-five terms in the expansions (\ref{dfffff}) and~(\ref{bfffff}) add up to a constant, leaving  
\begin{multline*}
T \;=\; 2 \iint F \, F\, dF_X \, dF_Y \;+\; \int F \, F \, dF \\ \;-\; 2 \int F_X \, F_Y \, F \, dF \;-\; \frac{1}{9}
\end{multline*}
By counting $F$'s in these integrals, one can represent $T = \int\!\!\int\!\!\int\!\!\int \vartheta \, dFdFdFdF$, where $\vartheta$ depends only on the induced patterns of four points, as claimed. 

\citet*[][]{bergsma2010nonparametric} and \citet*[][]{bergsma2010consistent,bergsma2014consistent} independently derived this independence measure, and denoted it by~$\tau^{\star} = 12T$. It seems they were the first to utilize it for testing independence in practice, and its popularity has risen in the past decade. 

Testing with $T$ proceeds similar to Hoeffding's test. The test statistic $T_n$ is an estimator of $T$ from $n$ samples, defined by fourfold summation analogously to Equation~(\ref{dn}). The null distribution of $T_n/3$ tends to the same limit as for $R_n$ and $D_n$ \citep{nandy2016large, dhar2016study}.

The computation of~$T_n$ and our new $O(n \log n)$ time algorithm are discussed in Part~\ref{algo}.

\newcommand{\cross}[4]{\raisebox{-7.5pt}{\tikz[scale=0.15]{
\foreach \x/\y in {1/#1,2/#2,3/#3,4/#4} {
\draw[black,fill=black] (\x,\y) circle (0.2);}
\draw[densely dotted,line width=0.4] (0,2.5)--(5,2.5);
\draw[densely dotted,line width=0.4] (2.5,0)--(2.5,5);
}}}

The case for the Bergsma--Dassios--Yanagimoto test is supported by its simplicity, as it uses the ranking patterns of four points rather than five. \cite{bergsma2014consistent} discuss an analogy between $\tau^{\star}$ and Kendall's $\tau$ correlation coefficient. They offer another compelling interpretation, noting that a formula of this test statistic classifies the pattern occurrences in the induced permutation into two types.
\begin{equation*}
T_n \;=\; \frac{1}{\tbinom{n}{4}}\left( \tfrac{1}{18} \sum_{\sigma \in \mathcal{C}} \#\sigma \,-\, \tfrac{1}{36} \sum_{\sigma \in \mathcal{D}} \#\sigma \right)
\end{equation*}
where 
\begin{align*}
\mathcal{C} \;&=\; \left\{ 1234, 1243, 2134, 2143, 3412, 3421, 4312, 4321 \right\} \\
\mathcal{D} \;&=\; S_4 \setminus \mathcal{C}
\end{align*}
are the sets of \emph{concordant} and \emph{discordant} 4-patterns, respectively. This classification of four-point rankings is best explained graphically as follows, by means of their \emph{cruciform partition}:
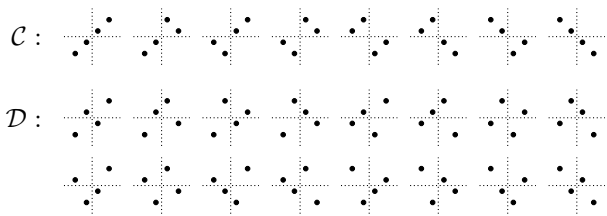
\begin{align*}
\mathcal{C}:& \;\;\; \cross1234 \;\, \cross1243 \;\, \cross2134 \;\, \cross2143 \;\, \cross3412 \;\, \cross3421 \;\, \cross4312 \;\, \cross4321 \\[0.5em]
\mathcal{D}:& \;\;\; \cross1324 \;\, \cross1342 \;\, \cross1423 \;\, \cross1432 \;\, \cross2314 \;\, \cross2341 \;\, \cross2413 \;\, \cross2431 \\
&\;\;\; \cross3124 \;\, \cross3142 \;\, \cross3214 \;\, \cross3241 \;\, \cross4123 \;\, \cross4132 \;\, \cross4213 \;\, \cross4231 
\end{align*}
Thus, departure from independence manifests as less quadruples with one point in each corner, and more with two pairs in opposite corners. 

\paragraph{Remark}

The statistic $T_n$ also arises in the combinatorial context of \emph{quasirandom} permutations. These are deterministic sequences $\pi_n \in S_n$ with asymptotic properties that random permutations satisfy with probability approaching one. For example, any box $[a,b]\times[c,d] \subseteq [0,1]^2$ contains $(b-a)(d-c)n \pm o(n)$ points of the form $(i/n,\pi_n(i)/n)$. \cite{cooper2004quasirandom} gives several equivalent characterizations, including various notions of discrepancy, and properties of the Fourier coefficients of~$\pi_n$. Another one is $\#\sigma(\pi_n)/\tbinom{n}{k} \to 1/k!$ for every $k$-pattern~$\sigma$, for all $k \in \mathbb{N}$.

The latter quasirandom property relates to the characterization of independence by patterns due to \cite{hoeffding1948non}. \cite{hoppen2011testing,hoppen2013limits} similarly use $\#\sigma(\pi_n)$ to define more general \emph{limit permutations}, that correspond to alternative distributions. As noted by \cite{even2018patterns},
the consistency of the Bergsma--Dassios--Yanagimoto statistic translates to the following simple criterion:
$$ \left\{\pi_n\right\} \text{ is quasirandom} \;\;\;\; \Leftrightarrow \;\;\;\; T_n\left(\pi_n\right) \to 0 $$
A series of papers on quasirandom and other limit permutations have independently reached the same conclusion \citep{kral2013quasirandom,glebov2015finitely,chan2019characterization}.

\section{ALGORITHMS}\label{algo}

We now present our algorithms for computing the consistent variants $R_n$ and $T_n$ of Hoeffding's statistic. The linear relation $T = D + 2R$ between these population coefficients implies that also for the sample statistics,
$$ T_n \;=\; D_n + 2R_n $$
See \cite{even2018patterns}
and \cite{drton2018high}. Since Hoeffding's $D_n$ is computed in time $O(n \log n)$, as shown in \S\ref{hoefdn}, A linear time computation of~$R_n$ would reduce to that of~$T_n$.

By definition $T_n$ can be computed in~$O(n^4)$, though \cite{bergsma2010consistent} noted that it can be improved to~$O(n^3)$ and left its complexity as an open problem. Later, \cite{bergsma2014consistent} suggested to approximate this statistic by averaging a random subset of the $\tbinom{n}{4}$ terms in its representation as a U-statistic. However, under the null hypothesis $T_n$ is degenerate of rank two, hence this approximation incurs a substantial loss of information even if as much as $O(n^2)$ terms are taken \citep{janson1984asymptotic}. \cite{weihs2016efficient} proposed an algorithm that precisely computes $T_n$ in $O(n^2 \log n)$ time and near linear memory use. \cite{heller2016computing} showed it can be computed in $O(n^2)$ time and memory. 

Also the apparent quadratic order cost of $R_n$ seems to have been regarded as a natural barrier, since it is an empiric version of the double integral $dF_X dF_Y$ defining~$R$,  compared to the single integral $dF$ in~$D$, see (\ref{dfffff}) and~(\ref{bfffff}). In fact, Deheuvels, Peccati and Yor have written that the fully consistent $R_n$ is less popular than $D_n$ only because it requires the summation of $n^2$ terms \citep[p.~496]{deheuvels2006quadratic}. Recently, \cite{chatterjee2020new} has considered near linear computational
time as the major concern when selecting a bivariate independence test in practice.

Our description of the near linear algorithm is completely self-contained, despite being inspired by the general method of \cite{even2019counting}.
It uses as a data structure a numeric array of fixed size~$n$, named \textsc{sum-array}, that supports assigning a value to a given position in $O(\log n)$ time, and also supports range sum queries in $O(\log n)$ time. Such an array is easy to implement with a complete binary tree of depth $\left\lceil\log n\right\rceil$, and goes back at least to \cite{shiloach1982n2log}. If $A$ is a \textsc{sum-array}, then $A.\textsc{prefix-sum}(y)$ returns $A[1] + \dots + A[y]$ in logarithmic time, and similarly for $A.\textsc{suffix-sum}$. 

Given the $n$ samples $(X_1,Y_1),\dots,(X_n,Y_n)$, we first sort them in $O(n\log n)$ time and compute the ranking permutation $\pi:\{1,\dots,n\}\to\{1,\dots,n\}$, satisfying $\pi(\text{rank}X_i) = \text{rank}Y_i$ as explained above in \S\ref{setting}. The inverse permutation $\pi^{-1}$ is clearly computable in linear time. We then run Algorithm~\ref{alg} on $\pi$. We claim that its output is the statistic $\tau^\star(\pi)$ of \cite{bergsma2014consistent}, which is $12T_n$ in our notation. 

\begin{algorithm}[tb]
\caption{Compute the $\tau^*$ statistic}
\linespread{1.25}\selectfont
\label{alg}
\begin{algorithmic}[1]
\Function{quad}{permutation $\pi \in S_n$}
\State $N \leftarrow 0$
\State $A \leftarrow \textsc{sum-array}(0,\dots,0)$
\Comment{of size $n$}
\State $A_u \leftarrow \textsc{sum-array}(0,\dots,0)$
\State $A_d \leftarrow \textsc{sum-array}(0,\dots,0)$
\State $A_{ud} \leftarrow \textsc{sum-array}(0,\dots,0)$
\For{$x$ \textbf{in} $(1,\dots,n)$}
\State $N_u \leftarrow A.\textsc{prefix-sum}(\pi[x])$
\State $N_d \leftarrow A.\textsc{suffix-sum}(\pi[x])$
\State $N_{du} \leftarrow A_d.\textsc{prefix-sum}(\pi[x])$
\State $N_{ud} \leftarrow A_u.\textsc{suffix-sum}(\pi[x])$
\State $N_{udu} \leftarrow A_{ud}.\textsc{prefix-sum}(\pi[x])$
\State $A[\pi(x)] \leftarrow 1$
\State $A_u[\pi(x)] \leftarrow N_u$
\State $A_d[\pi(x)] \leftarrow N_d$
\State $A_{ud}[\pi(x)] \leftarrow N_{ud}$
\State $\Delta \leftarrow 2N_{udu} - N_{du} N_d - N_{ud}N_u + (x-2) N_u N_d$
\State $N \leftarrow N + \Delta$
\EndFor
\State \Return $N$
\EndFunction
\vspace{0.5em}
\Function{tau-star}{permutation $\pi \in S_n$}
\State $S \leftarrow \textsc{quad}(\pi)$
\State $S \leftarrow S + \textsc{quad}(\pi(n),\dots,\pi(1))$
\State $S \leftarrow S + \textsc{quad}(\pi^{-1}(1),\dots,\pi^{-1}(n))$
\State $S \leftarrow S + \textsc{quad}(\pi^{-1}(n),\dots,\pi^{-1}(1))$
\State \Return $\frac23 - \frac{1}{4}\, S / \binom{n}{4}$
\EndFunction
\end{algorithmic}
\end{algorithm}

\begin{theorem}
\label{algworks}
Algorithm~\ref{alg} returns $\tau^{\star}$.
\end{theorem}

The proof is given in the supplementary material, \S\ref{algproof}. Here are some remarks on the running time of this algorithm.
\begin{enumerate}
\item 
The statements on lines 8-16 are called $4n$ times and run in $O(\log n)$ time. All the rest runs in~$O(n)$. In total, the algorithm executes in $O(n \log n)$ time.
\item 
Computing $\tau^{\star}$ by the scheme described in \cite{even2019counting}
would require 32 equivalents of the subroutine \textsc{quad}, rather than 4 in this dedicated algorithm. This yields a 8x speedup.
\item
As with the previous algorithms for this problem in the literature, we work under the assumption that arithmetic operations on numbers of $\log n$ digits cost~$O(1)$ time, where $n$ is the input length.
\end{enumerate}

\section{R PACKAGE}
\label{rpackage}

We have implemented the above algorithm in \texttt{C++}, and have made it available in a new \texttt{R} package, named \texttt{independence} \citep[on CRAN,][]{independenceR}.
Our package provides methods for near linear time computation of the following three statistics.
\begin{itemize}[itemsep=0.25em,topsep=0pt]
\item Hoeffding's $D_n$
\item The refined Hoeffding statistic $R_n$
\item Bergsma--Dassios--Yanagimoto's $\tau^{\star} = 12T_n$
\end{itemize}
Run times of these three tests on large-scale datasets are summarized in Table~\ref{table}. These experiments were performed using Intel 3.30GHz CPU and less than 8GB RAM.

\begin{table}[h]
\caption{Run Times for Independence Tests (sec)} \label{table}
\begin{center}
\begin{tabular}{crrr}
\textbf{$n$} & \textbf{Hoeffding} & \textbf{Refined} & \textbf{Tau-Star} \\
\hline \\
$10^5$ & 0.16 \;& 0.23 \;& 0.22 \;\\
$10^6$ & 0.60 \;& 3.51 \;& 3.38 \;\\
$10^7$ & 6.74 \;& 58.86 \;& 55.42 \;\\
$10^8$ & 85.74 \;& 888.75 \;& 805.75 \;\\
\end{tabular}
\end{center}
\end{table}

For computing $p$-values, if needed, we call the method \texttt{pHoeffInd} from the \texttt{R} package \texttt{TauStar} by \cite{TauStar}. Look at the documentation of the package \texttt{independence} for technical details about $p$-value precision, some limitations due to the machine's word size, and handling ties and missing values.

Figure~\ref{code} shows
a short code excerpt using the package. It illustrates how the dependency given by \cite{yanagimoto1970measures}, as described in~\S\ref{fooling}, is undetected by Hoeffding's original test but clearly visible to its refined variants.

\begin{figure}[t]
\centering
\begin{tcolorbox}
\begin{verbatim}
> library(independence)
> set.seed(12345)
> f <- function(a,b) ifelse(a>b, 
+       pmin(b,a/2), pmax(b,(a+1)/2))
> x <- runif(300)
> y <- f(x, runif(300))
> hoeffding.D.test(x,y)$p.value
[1] 0.4589397
> hoeffding.refined.test(x,y)$p.value
[1] 2.138784e-10
> tau.star.test(x,y)$p.value
[1] 1.053099e-08
\end{verbatim}
\end{tcolorbox}
\vspace{.3in}
\caption{Using the \texttt{R} Package \texttt{independence} where Hoeffding's $D$ Fails.}
\vspace{.3in}
\label{code}
\end{figure}

\section{RELATED TESTS}
\label{related}

Ever since Hoeffding introduced his~$D$, methods of testing and quantifying dependence have continued to be developed. Many competing ideas for other population coefficients of dependence and test statistics have emerged in the statistical literature, often driven by the need to go beyond Hoeffding's original setting of two real-valued variables. Several tests have been proposed for \emph{multidimensional} variables, $X \in \mathbb{R}^p$ and~$Y \in \mathbb{R}^q$, as well as in more general spaces. For example, the \emph{distance correlaiton} (dCor) by \citet{szekely2007measuring,lyons2013distance} is based on pairwise distances, and another approach by~\citet{gretton2005measuring,gretton2008kernel} relies on a Hilbert--Schmidt kernel structures associated to the two spaces (HSIC), and other recent approaches \citep[e.g.][]{heller2013consistent,deb2019multivariate,shi2020rate}. Another extension is to more than two variables. Indeed, if $X,Y,Z$, etc.~are pairwise independent, it is still interesting to test whether they are \emph{mutually independent}. In the problem of \emph{conditional independence}, one asks whether $X$ and $Y$ are independent under conditioning on~$Z$ \citep[for example]{zhang2011kernel}. However, we leave such generalizations for another time, and remain within the realm of \emph{rank-based independence testing for two real variables}, as even in this specific setting of Hoeffding's test, there are several other approaches to discuss.

Recall that we only assume that $X$ and $Y$ are real-valued and non-atomic, and $n$ iid samples $(X_i,Y_i)$ are given. Nothing is assumed about their joint or marginal distribution, though it would also make sense to discuss also the ``copula'' setting where $F_X$ and $F_Y$ are known, or only one of them, or just known to be the same, and other assumptions on the marginal. In our setting, a distribution-free statistic only depends on the ranking permutation, $\pi \in S_n$ with $\pi(\mathrm{rank}\,X_i) = \mathrm{rank}\,Y_i$. Under the null hypothesis, every $\pi \in S_n$ is equaly likely. Therefore, our independence measures may be viewed as ways to order $S_n$ such that permutations that are more likely under dependent alternatives tend to appear near the end.

We list a selection of examples for other approaches to independence testing depending only on ranks.

\begin{itemize}
\itemsep0.05em
\item Hoeffding's dependence measure $D$ is an $L^2$ norm of $(F_{XY} - F_XF_Y)$, and so is the statistic $D_n$ with empirical cdfs. A~Kolmogornov--Smirnov type statistic uses the maximum norm instead \citep{blum1961distribution, deheuvels1979fonction, Hmisc}. 
\item
Other ideas along these lines: introducing weight functions \citep{rosenblatt1975quadratic,de1980cramer}, using $L^1$ distance \citep{Hmisc}, or comparing the empirical characteristic functions of $F_{XY}$ and $F_XF_Y$ \citep{feuerverger1993consistent,csorgHo1985testing}.
\item 
It is suggested by \cite{szekely2009brownian} to apply their \emph{distance correlation} (dCor) on ranks. This means using it with the points $\{(x,\pi(x))\}_{x=1}^{n}$ rather than the original samples.
\item 
From another viewpoint, Hoeffding's formula in \S\ref{hoefdn} averages Pearson's $\chi^2$ over all partitions of the data into $2 \times 2$ contingency tables: $\tfrac{a_i}{c_i}\!|\!\tfrac{b_i}{d_i}$. \citet{heller2016consistent} study a scheme of tests that take more partitions, finer partitions, and other ways to weigh the scores (HHG).\item
Another test that uses binning is the Maximal \emph{Mutual Information Coefficient} (MIC) by \cite{reshef2011detecting}, and its variations \citep{reshef2013equitability,gorfine2012comment,simon2014comment,renyi1959measures}
\item
\cite{chan2019characterization} found more combinations of 4-patterns that consistently detect independence, similar to $\mathcal{C}$ for $\tau^\star$, in \S\ref{bdy}. As noted by \cite{even2019counting},
the following two are computable in linear time.
\begin{align*}
&\{1324, 1342, 2413, 2431, 3124, 3142, 4213, 4231\} \\
&\{1324, 1423, 2314, 2413, 3142, 3241, 4132, 4231\}
\end{align*}
In general, such combinations form a convex cone of consistent independence measures. 
\item 
A new quick measure $\xi_n$ by \cite{chatterjee2020new} uses $\sum_{x=1}^{n-1}|\pi(x+1)-\pi(x)|$. See also \cite{shi2020power,wang2015efficient}.
\item 
\cite{garcia2014independence} suggested to use the longest monotone subsequence in~$\pi$. 
\end{itemize}

\subsubsection*{Acknowledgements}

The author was supported by the Lloyd’s Register Foundation / Alan Turing Institute programme on Data-Centric Engineering.

\bibliography{independence}


\onecolumn
\appendix 
\section{PROOF OF PROPOSITION \ref{fool2}}\label{foolproof}

For $X = 0.X_1X_2X_3\dots$ with independent $X_i \sim \text{Ber}(\tfrac12)$, it is a standard fact that $X \sim U(0,1)$. Similarly for $Y = 0.Y_1Y_2Y_3\dots Y_\ell X_{\ell+1}X_{\ell+2}\dots$, since given only the number~$\ell(X)$, the next digit $Y_{\ell+1} = X_{\ell+1} \sim \text{Ber}(\tfrac12)$ and all the other digits of $Y$ are independent by definition. Hence also $Y \sim U(0,1)$ without conditioning on~$\ell$.

It is hence left to show that the joint cdf $F(x,y)=xy$ almost everywhere with respect to~$F$. Let $(x,y)$ be any point obtained from the random construction. The distribution is clearly supported on the set of such points. Suppose $x<0.5$, as the case $x>0.5$ follows by symmetry. 

In order to compute $F(x,y) = P(X \leq x, Y \leq y)$, we divide into cases according to certain ranges for~$X$. It will be convenient to represent $x$ as a series:
$$ x \;=\; 2^{-a} + 2^{-b} + 2^{-c} + 2^{-d} + \dots $$
for some $1 < a < b < c < d < \cdots$. Note that this expansion is infinite with probability one. Denote its ``tails'' by $x_b = 2^{-b} + 2^{-c} + \dots$, and $x_c = 2^{-c} + 2^{-d} + \dots$, and so on. This yields a partition of the interval $[0,x]$:
$$ x \,=\, x_a \,>\, 2^{-a} \,>\, x_b \,>\, 2^{-b} \,>\, x_c \,>\, 2^{-c} \,>\, x_d \,>\, 2^{-d} \,>\, \cdots \,>\, 0 $$
The proof prooceeds by computing the conditional probability that $Y \leq y$, given the interval where~$X$ falls. 

Suppose that $X \in (2^{-i}, x_i)$ for some~$i$. This means that the binary expansion of~$X$ starts with a run of exactly $i-1$ zeros, and then its $i$th digit is one. Since $Y$ is obtained by randomizing these zeros, $Y = N/2^{i-1} + X$ with a uniformly random integer $N \in \{0,1,\dots,2^{i-1}-1\}$. Similarly $y = n/2^{i-1} + x_i$ for some integer $n$, because $(x,y)$ is obtained from the same process. It follows that $Y < y$ if and only if $N \leq n$, where the case of equality follows from~$X < x_i$. Since the $2^{i-1}$ values for~$N$ are equally likely, 
$$ P\left(Y \leq y \;|\;  2^{-i} < X < x_i \right) \;=\;  \frac{n+1}{2^{i-1}} \;=\; y - x_i + 2^{-(i-1)} $$

Now suppose that $X \in (x_j,2^{-i})$ for some $i<j$, corresponding to adjacent terms in the above series representation of~$x$. As before, $y = n/2^{i} + x_j$ for some integer $n$, and $N \in \{0,1,\dots,2^i-1\}$ denotes the uniformly random contribution of the first $i$ digits of~$Y$. Apparently, we only have a lower bound $Y \geq N/2^i + x_j$ since \emph{at least} $i$ zeros are resampled. However, clearly $Y < (N+1)/2^i$. This is enough to deduce that $N<n$ if and only if $Y < y$, and therefore,
$$ P\left(Y \leq y \;|\;  x_j < X < 2^{-i} \right) \;=\; \frac{n}{2^i} \;=\; y - x_j $$

For the whole interval $(x_j,x_i)$, we use the law of total probability and the two previous computations, and then simplify using $x_i - x_j = 2^{-i}$.
\begin{align*}
P\left(X \in (x_j,x_i), \; Y \leq y\right) \;&=\; P\left(X \in (x_j, 2^{-i}), \; Y \leq y\right) \,+\, P\left(X \in (2^{-i},x_i), \; Y \leq y\right)
\\ \;&=\; (2^{-i}-x_j)(y-x_j) + (x_i-2^{-i})(y - x_i + 2^{-(i-1)}) \\
\;&=\; (2^{-i}-x_j)(y-x_j) + x_j(y-x_j + 2^{-i}) \\
\;&=\; 2^{-i} y
\end{align*}
Finally, by summation over intervals for $X$,
\begin{align*}
P\left(X \leq x, \; Y \leq y\right) \;&=\; P\left(X \in (x_b, x_a), \; Y \leq y\right) \,+\, P\left(X \in (x_c, x_b), \; Y \leq y\right) \,+\, \dots \\ 
\;&=\; 2^{-a} y \;+\; 2^{-b} y \;+\; 2^{-c} y \;+\; \dots \\ 
\;&=\; xy
\end{align*}
In conclusion, $F(x,y) = F_X(x)F_Y(y)$ for almost every $(x,y)$ with respect to~$F$, and Hoeffding's $D=0$ by the definition.
\qed

\section{PROOF OF PROPOSITION \ref{algworks}}\label{algproof}

Our proof provides a step-by-step human-verifiable analysis of the functions \textsc{quad} and \textsc{tau-star}. We note that the general machinery in \cite{even2019counting} can automate some this work.

\begin{lemma}
\label{quad}
Let $\pi \in S_n$. Then the function \textsc{quad} in Algorithm~\ref{alg} computes the following combination of pattern occurrence counts in~$\pi$:
$$ 2\,\#1324 + \#1342 + 3\,\#1423 + 2 \#1432 - \#2143 +2\,\#2314 +2\,\#2413 + \#3412 +2\,\#4123 + \#4132 + \#4213 $$
\end{lemma}

\begin{proof}
Denote the values of the variables in the function \textsc{quad}, as assigned during the $x$th iteration of the \textbf{for} loop on lines~7-18, by 
$$ N_u^{(x)},\,N_d^{(x)},\,N_{ud}^{(x)},\,N_{du}^{(x)},\,N_{udu}^{(x)},\,A^{(x)},\,A_u^{(x)},\,A_d^{(x)},\,A_{ud}^{(x)},\,\Delta^{(x)} \;\;\;\;\;\;\;\; 1 \leq x \leq n $$  

The following chain of assertions relate these variable to pattern occurrences in $\pi$, which are counted during the $x$th iteration.
\begin{itemize}
\item
By the assignment on line 13, $A^{(x)}[y] = 1$ if $y=\pi(x')$ for some $x' \leq x$, and 0 otherwise. By line 8, $N_u(x)$ is the sum of $A^{(x)}[y]$ for $y < \pi(x)$. Hence it is the number of positions $x'$ such that $x'<x$ and $\pi(x')<\pi(x)$. In other words, this is the number of increasing pairs of entries of $\pi$ whose last position is $x$. We therefore express it as
$$ N_u^{(x)} \;=\; \#12^{(x)} $$
meaning the number of occurrences of the pattern $12$ whose last position is~$x$.
\item
Similarly, by line 9,
$$ N_d^{(x)} \;=\; \#21^{(x)} $$
\item
By line 14,
$A_u^{(x)}[y] = \#12^{(x')}$ if $x' = \pi^{-1}(y) \leq x$, and 0 otherwise.  $N_{ud}^{(x)}$ is the sum of $A_u^{(x)}[y]$ for all $y > \pi(x)$, by line 11. That is the number of increasing pairs that end at some $x' < x$, and then decrease such that $\pi(x') > \pi(x)$. In conclusion, this suffix sum counts $x''$ and $x'$ such that $x'' < x' < x$ and $\pi(x'') < \pi(x') > \pi(x)$. It can be expressed using the following two pattern counts:
$$ N_{ud}^{(x)} \;=\; \#132^{(x)} + \#231^{(x)} $$
\item
Similarly, By lines 15 and 10,
$$ N_{du}^{(x)} \;=\; \#213^{(x)} + \#312^{(x)} $$
\item
By line 16, $A_{ud}^{(x)}[y] = \#132^{(x')} + \#231^{(x')}$ for $x' = \pi^{-1}(y) \leq x$ and 0 otherwise. Then, on line 12, we sum over $y<\pi(x)$ and obtain the number of $x''' < x'' < x' < x$ such that $\pi(x''') < \pi(x'') > \pi(x') < \pi(x)$, or equivalently the following sum over so-called ``alternating'' patterns:
$$ N_{udu}^{(x)} \;=\; \#1324^{(x)}+\#1423^{(x)}+\#2314^{(x)}+\#2413^{(x)}+\#3412^{(x)}  $$
\item
The product of $N_u^{(x)}$ and $N_d^{(x)}$ is the number of $x',x'' < x$ with $\pi(x') > \pi(x)$ and $\pi(x'') < \pi(x)$. Such pairs of positions are exactly occurrences of the following two patterns:
$$ N_u^{(x)}  N_d^{(x)} \;=\; \#132^{(x)} + \#312^{(x)} $$
\item
Given $x'$ and $x''$ as above, there are $(x-3)$ other positions $x''' < x$ left. Considering the 6 cases for the ordering of $x',x'',x'''$ and the 4 cases for the ranking of $\pi(x''')$ among $\pi(x'')<\pi(x)<\pi(x')$, the number of such trios of positions $x',x'',x'''$ is a sum of 24 patterns counts:
\begin{align*}
(x-3)N_u^{(x)}  N_d^{(x)} \;&=\;  2\,\#1342^{(x)} + 2\,\#1432^{(x)} + 2\,\#3142^{(x)} + 2\,\#3412^{(x)} + 2\,\#4132^{(x)} + 2\,\#4312^{(x)} \\ &\,+\; 2\,\#1243^{(x)} + 2\,\#1423^{(x)} + 2\,\#2143^{(x)} + 2\,\#2413^{(x)} + 2\,\#4123^{(x)} + 2\,\#4213^{(x)}
\end{align*}
\item
The product of $N_{ud}^{(x)}$ and $N_u^{(x)}$ gives the number of $x',x'',x''$ such that $x'' < x' < x$ and $x''' < x$ while $\pi(x'') < \pi(x') > \pi(x)$ and $\pi(x''') < \pi(x)$. Here there are 3 possible orderings for $x',x'',x'''$ and 3 for $\pi(x'),\pi(x''),\pi(x''')$ and one more case where $x''$ and $x'''$ coincide:
$$ N_{ud}^{(x)} N_u^{(x)} \;=\; \#1342^{(x)} + \#3142^{(x)} + \#3412^{(x)} + 2\;\#2143^{(x)} + 2\;\#1243^{(x)} + \#1423^{(x)} + \#2413^{(x)} + \#132^{(x)} $$
\item
Similarly,
$$ N_{du}^{(x)} N_d^{(x)} \;=\; \#4213^{(x)} + \#2413^{(x)} + \#2143^{(x)} + 2\;\#3412^{(x)} + 2\;\#4312^{(x)} + \#4132^{(x)} + \#3142^{(x)} + \#312^{(x)} $$
\end{itemize}
Using the last five pattern combinations and line 17, one can express $\Delta^{(x)}$ with patterns as well. We sum $\Delta^{(x)}$ over $x \in \{1,\dots,n\}$, on line 18. Since for every pattern $\sum_{x=1}^n\#\sigma^{(x)}=\#\sigma$, the summation is done by dropping the superscript ``$(x)$'' from all pattern counts, as follows. 
\begin{align*}
&\text{\textsc{quad}}(\pi) \;=\; \sum_{x=1}^n \Delta^{(x)} \;=\; \sum_{x=1}^n \left[ 2N_{udu}^{(x)} - N_{du}^{(x)} N_d^{(x)} - N_{ud}^{(x)}N_u^{(x)} + (x-3) N_u^{(x)} N_d^{(x)} + N_u^{(x)} N_d^{(x)} \right] \\
&= 2\,\#1324 + 2\,\#1423 + 2\,\#2314 + 2\,\#2413 + 2\,\#3412 \\
&\;\;\;\; - \#4213 - \#2413 - \#2143 - 2\;\#3412 - 2\;\#4312 - \#4132 - \#3142 - \#312 \\
&\;\;\;\; -\#1342 - \#3142 - \#3412 - 2\;\#2143 - 2\;\#1243 - \#1423 - \#2413 - \#132 \\
&\;\;\;\; +2\,\#1342 + 2\,\#1432 + 2\,\#3142 + 2\,\#3412 + 2\,\#4132 + 2\,\#4312 \\ 
&\;\;\;\; + 2\,\#1243 + 2\,\#1423 + 2\,\#2143 + 2\,\#2413 + 2\,\#4123 + 2\,\#4213 \\
&\;\;\;\; +\#132+\#312 \\[0.5em]
&= 2\,\#1324 + \#1342 + 3\,\#1423 + 2\,\#1432 - \#2143 + 2\,\#2314 + 2\,\#2413 + \#3412 + 2\,\#4123 + \#4132 + \#4213 
\end{align*}
This is the combination stated in the lemma.
\end{proof}

We now analyze the main function \textsc{tau-star}. Given a permutation $\pi = (\pi(1),\dots,\pi(n)) \in S_n$, we denote its two reflections, \emph{reverse} and \emph{inverse}, as follows: $\text{rev}\,\pi = (\pi(n),\dots,\pi(1))$ and $\text{inv}\,\pi = (\pi^{-1}(1),\dots,\pi^{-1}(n))$. It is easy to see that pattern occurrence respects these symmetries, in the sense that $\#\sigma(\pi) = \#\,\text{rev}\,\sigma\,(\text{rev}\,\pi)$ and $\#\sigma(\pi) = \#\,\text{inv}\,\sigma\,(\text{inv}\,\pi)$.

On lines 21-24, the subroutine \textsc{quad} is called four times with the inputs $\pi$, $\text{rev}\,\pi$, $\text{inv}\,\pi$, and $\text{rev}\,\text{inv}\,\pi$, respectively. The results of these four calls can  be expressed as pattern counts in the original~$\pi$, by applying $\text{rev}$, $\text{inv}$, and $\text{inv}\circ\text{rev}$ to the patterns in Lemma~\ref{quad}. This gives the following four expressions:
\begin{align*}
& 2\,\#1324 + \#1342 + 3\,\#1423 + 2 \#1432 - \#2143 +2\,\#2314 +2\,\#2413 + \#3412 +2\,\#4123 + \#4132 + \#4213 \\
& 2\,\#4231 + \#2431 + 3\,\#3241 + 2 \#2341 - \#3412 +2\,\#4132 +2\,\#3142 + \#2143 +2\,\#3214 + \#2314 + \#3124 \\
& 2\,\#1324 + \#1423 + 3\,\#1342 + 2 \#1432 - \#2143 +2\,\#3124 +2\,\#3142 + \#3412 +2\,\#2341 + \#2431 + \#3241 \\
& 2\,\#4231 + \#4132 + 3\,\#4213 + 2 \#4123 - \#3412 +2\,\#2431 +2\,\#2413 + \#2143 +2\,\#3214 + \#3124 + \#2314 
\end{align*}
The sum of these lines is
\begin{align*}
S \;= &\;4\,\#1324+4\,\#1342+4\,\#1423+4\,\#1432+4\,\#2314+4\,\#2341+4\,\#2413+4\,\#2431 \\
&\;4\,\#3124+4\,\#3142+4\,\#3214+4\,\#3241+4\,\#4123+4\,\#4132+4\,\#4213+4\,\#4231
\end{align*}
Note that these are exactly the 16 discordant patterns, denoted by $\mathcal{D}$ in~\S\ref{bdy}. In conclusion,
$$ \tau^{\star} \;=\; 12T_n \;=\; \frac{1}{\binom{n}{4}}\left( \frac{2}{3} \sum_{\sigma \in \mathcal{C}} \#\sigma \,-\, \frac{1}{3} \sum_{\sigma \in \mathcal{D}} \#\sigma \right) \;=\; \frac{1}{\binom{n}{4}}\left( \frac{2}{3} \sum_{\sigma \in S_4} \#\sigma \,-\, \sum_{\sigma \in \mathcal{D}} \#\sigma \right) \;=\; \frac23 \,-\, \frac{S}{4}/\binom{n}{4}
$$
as returned in line 25.
\qed

\end{document}